\documentclass[12pt]{article}
\usepackage{amsmath,amsthm}
\usepackage{amssymb}
\usepackage{lmodern}
\usepackage{enumerate}
\usepackage{graphicx}
\usepackage[usenames,dvipsnames]{xcolor}
\usepackage[unicode=true,pdfusetitle,
bookmarks=true,bookmarksnumbered=true,bookmarksopen=true,bookmarksopenlevel=1,
breaklinks=false,pdfborder={0 0 
    0},backref=false,colorlinks=true,linkcolor=NavyBlue,citecolor=NavyBlue,urlcolor=black]
{hyperref}
\pdfpageheight\paperheight
\pdfpagewidth\paperwidth

\usepackage[capitalise]{cleveref}
\crefformat{equation}{(#2#1#3)}

\numberwithin{equation}{section}
\newtheorem{theorem}{Theorem}[section]
\newtheorem{lemma}[theorem]{Lemma}
\newtheorem{proposition}[theorem]{Proposition}

\usepackage{dsfont}

\newcommand{\mc}[1]{\mathcal{#1}}

\newcommand{\R}{\mathbb{R}}
\renewcommand{\C}{\mathbb{C}}
\newcommand{\Z}{\mathbb{Z}}
\newcommand{\ind}{\mathds{1}}
\renewcommand{\epsilon}{\varepsilon}
\newcommand{\Exterior}{\mathchoice{{\textstyle\bigwedge}}%
    {{\bigwedge}}%
    {{\textstyle\wedge}}%
    {{\scriptstyle\wedge}}}

\DeclareMathOperator{\pr}{\mathbb{P}}
\DeclareMathOperator{\ex}{\mathbb{E}}
\DeclareMathOperator{\var}{Var}

\DeclareMathOperator{\diag}{diag}
\DeclareMathOperator{\dist}{dist}
\DeclareMathOperator{\spec}{spec}

\DeclareMathOperator{\Ran}{Ran}
\DeclareMathOperator{\Ker}{Ker}

\DeclareMathOperator{\rank}{rank}
\newcommand{\norm}[1]{\left\|#1\right\|}
\begin{document}

\title{On the Sum of the Non-Negative Lyapunov Exponents for Some Cocycles Related to the Anderson Model}

\author{Ilia Binder, Michael Goldstein, Mircea Voda}

\date{}
\maketitle

\begin{abstract}
	We provide an explicit lower bound for the the sum of the non-negative Lyapunov exponents for some cocycles related to 
	the Anderson model. In particular, for the Anderson model on a strip
	of width $ W $ the lower bound is  proportional to $ W^{-\epsilon} $, for any $ \epsilon>0 $. This bound is 
	consistent with the fact that the lowest non-negative Lyapunov exponent is conjectured to have a lower bound 
	proportional to $ W^{-1} $.
\end{abstract}

\tableofcontents

\section{Introduction}

We consider the cocycles associated with a family of random operators on the strip
$ \Z_W:=\Z\times \{ 1,\ldots,W \} $ defined by 
\begin{equation}\label{eq:strip-Anderson}
	(H\Psi)_n = -\Psi_{n-1}-\Psi_{n+1}+S_n\Psi_n,
\end{equation}
where $ \Psi\in l^2(\Z,\C^W)\equiv l^2(\Z_W) $, and
\begin{equation*}
	S_n:=\diag(V_{(n,1)},\ldots,V_{(n,W)})-U_n.
\end{equation*}
The potentials $ V_i $ are i.i.d. random variables and the random matrices $ U_n $ are symmetric and independent of  $ V_i $. 
Let $ d $  be a positive integer such that for any $ n $ we have $ U_n(x,y)=0 $ for $ |x-y|> d $. Our methods work even when 
$ d=W $, but the results are close to optimal only when $ d $ is fixed as $ W\to \infty $. For this reason, the results will emphasize
only the latter case and the dependence on $ d $ won't be stated explicitly, but it will be clear in the proofs. 
In particular, if
\begin{equation*}
	U_n(x,y)=\begin{cases}
		1 &, \text{ if } |x-y|=1\\
		0 &, \text{ otherwise}
	\end{cases},
\end{equation*}
then we obtain the usual tight-binding Anderson model on the strip $ \Z_W $.
In this paper we assume that $ V_i $ have bounded density function $ \rho $ and we let
\begin{equation}\label{eq:intro-densityub}
  D_0 := \sup_{x} \rho(x) < + \infty.
\end{equation}
Furthermore we assume a weak integrability condition:
\begin{equation}\label{eq:intro-densitydecay}
  \pr(|V_i|\ge T) \le D_1/T,~\pr(\norm{U_n}\ge T) \le D_1/T,~T\ge 1.
\end{equation}

 The $ 2W\times 2W $ matrix 
\begin{equation*}
	T_N^E=\prod_{k=N}^{1}\begin{bmatrix}
		S_k-E & -I \\ I & 0
	\end{bmatrix}
\end{equation*}  
is called the $ N $-step transfer matrix and satisfies
\begin{equation}\label{eq:T_N-recurrence}
	\begin{bmatrix}
		\Psi_{N+1} \\ \Psi_{N}
	\end{bmatrix}
	=
	T_N^E \begin{bmatrix}
		\Psi_1 \\ \Psi_0
	\end{bmatrix},
\end{equation}
for any solution of the difference equation $ H\Psi=E\Psi $.
Let $ \gamma_1^E\ge \ldots \ge \gamma^E_W\ge 0 $ be the non-negative Lyapunov exponents associated with 
the cocycle determined by $ T_N^E $. See \cite[III.5,IV.3]{BL-85-Products} for 
definitions and further information. Note that since $ T_N^E $ are symplectic matrices it follows that 
the Lyapunov spectrum of the cocycle splits into a non-negative and a non-positive part, symmetrically with 
respect to zero. Furthermore, from this and the simplicity of the Lyapunov spectrum (see 
\cite{GM-87-condition}) it follows that in fact $ \gamma_1^E>\ldots >\gamma_W^E>0 $. The positivity of 
$ \gamma_W^E $ is crucial for Anderson localization. Indeed, it is known that the quasi-one-dimensional
Anderson model has (almost surely) only pure point spectrum and its eigenvectors decay exponentially with 
decay rate proportional to $ \gamma_W^E $ (see \cite{KLS-90-Localization}). 
It is thus natural to ask for an explicit lower bound on $ \gamma_W^E $. In particular one is interested in
the asymptotics of the lower bound when $ W\to \infty $. This is motivated by the idea of using the localization
on strips to understand what happens in $ \Z^2 $.   

The best known lower bound for $ \gamma_W^E $, due to Bourgain \cite{Bou-13-lower} (in the case of the Anderson
model), 
is $ \exp(-CW(\log W)^4) $; however, it is expected that it should be 
of the order of $ W^{-1} $. In \cite{BGV-13-fluctuations} we argued that it is possible to obtain lower bounds
of the order of $ W^{-C} $ provided that we have have a good enough lower bound on the variance of Green's
function. The estimate we obtained for the fluctuations of the Green's function was far from optimal and only 
yielded a somewhat weaker estimate for $ \gamma_W^E $ than the one from \cite{Bou-13-lower} . The idea of estimating the localization 
length through the fluctuations of the resolvent has been previously implemented by Schenker 
\cite{Sch-09-Eigenvector} in the context
of random band matrices, but it is not clear how to adjust his developments to the Anderson model. 

In this paper we consider a problem for which the approach of \cite{BGV-13-fluctuations} yields better results.
Namely, we provide a lower bound for the sum of the non-negative Lyapunov exponents through a lower bound on 
the fluctuations of the Dirichlet determinants.  We let $ H_N $ be the restriction of $ H $ to 
$ [1,N]\times[1,W] $ with Dirichlet boundary
conditions, and we denote the characteristic polynomial $ \det(H_N-E) $ by
$ f_N^E $. Our main result is as follows.

\begin{theorem}\label{thm:lower-bound-sum-Lyapunov}
	Let $ E\in \R $ and assume that there exists $ \beta_0>0 $  such that $ \var(\log|f_N^E|)\ge \beta_0 NW $ 
	for any $ N\ge 1 $. Then for any $ \epsilon>0 $ there exists
	a constant
	\begin{equation*}
		c_0=c_0(D_0,D_1,\beta_0,|E|,\epsilon)
	\end{equation*}
	such that
	\begin{equation*}
		\gamma^E_1+\ldots+\gamma^E_W\ge c_0 W^{-\epsilon}.
	\end{equation*}
\end{theorem}
Note that the estimate we obtain is consistent with the conjectured estimate $ \gamma_W^E\gtrsim W^{-1} $.
We believe 
that the estimate for the fluctuations of $ \log|f_N^E| $ should be true for general potentials, but we can only establish 
it for potentials with ``considerable tails''. Given a set 
$ \Lambda\subset \Z_W $ we let $ H_\Lambda $ be the restriction, with Dirichlet boundary conditions, of $ H $ 
to $ \Lambda $ and we let $ f_\Lambda^E=\det(H_\Lambda-E) $. 
\begin{theorem}\label{thm:var-lower_bound}
	Let $ E\in \R $ and $ \Lambda\subset\Z_W $. Then there exist constants $ C_0=C_0(D_0,D_1,|E|,d) $ and $ C_1 $ such that
	\begin{equation*} 
		\var(\log|f_\Lambda^E|)\gtrsim |\Lambda| |I| \inf_I \rho(x),
	\end{equation*}
	for any interval $ I $ of the form $ [M_0,C_1M_0] $ or $ [-M_0,-C_1M_0] $, with $ M_0\ge C_0 $.
\end{theorem} 
Note that in this paper the constants implied by symbols like $ \gtrsim $ or $ \gg $ will always be
absolute constants.  The above theorem says that if $ \rho $ has large enough tails, that is if $ \rho>0 $
on a large enough interval $ I $, then 
\begin{equation*}
	\var(\log|f_\Lambda^E|)\ge c(\rho,I,|E|,d)|\Lambda|	
\end{equation*}
and, in particular, the assumption needed for 
\cref{thm:lower-bound-sum-Lyapunov} is satisfied. Of course, for the lower bound in \cref{thm:lower-bound-sum-Lyapunov} to be 
proportional to $ W^{-\epsilon} $ it is important that $ d $ should be independent of $ W $. In the case $ d=d(W) $ it should be 
clear from the proofs that we obtain a lower bound proportional to $ W^{-C} $, with some absolute
constant $ C $.

The estimate we obtained on the fluctuations of $ \log|f_\Lambda^E| $ is probably optimal in $ \Z_W $. This should be clear from the following large deviations estimate for the Dirichlet determinants. 
\begin{theorem}\label{thm:ldt}
	Let $ \epsilon>0 $ and $ E\in \R $. Then there exists a constant $ C_0=C_0(d,D_0,D_1,|E|,\epsilon) $ such 
	that for any $ K\gg 1 $ and
 	any rectangular set $ \Lambda\subset \Z_W $ with $ |\Lambda|\ge C_0 $ we have
	\begin{equation*}
		\pr \left( \left|\log|f_\Lambda^E|-\ex \left( \log|f_\Lambda^E| \right)\right|
			> |\Lambda|^{1/2+\epsilon}K \right)\le \exp(-K/2).
	\end{equation*}
\end{theorem}
The large deviations estimate implies that $ \var(\log|f_\Lambda^E|)\lesssim |\Lambda|^{1+2\epsilon} $, for 
any $ \epsilon>0 $ and $ |\Lambda|\ge C_0 $. It follows that both the fluctuations lower bound and the large
deviations estimate are close to being optimal on $ \Z_W $. Note that the large deviations estimate is in fact independent of $ W $,
so it is really a result on $ \Z^2 $.

We will now discuss the proof of \cref{thm:lower-bound-sum-Lyapunov}. From the proof of the Thouless formula on the strip (see \cite{CS-83-Log}; cf. \cite{KS-88-Stochastic-Schr} and \cite[Prop. VI.4.4]{CL-90-Spectral}) 
we have that
\begin{equation}\label{eq:determinant-to-Lyapunov}
	\gamma_1^E+\ldots+\gamma_W^E=\lim_{N\to\infty}\frac{\ex (\log|f_N^E|)}{N}.
\end{equation}
In fact, as a consequence of the large deviations estimate from \cref{thm:ldt} it follows that we have the stronger pointwise result: 
\begin{equation*}
	\gamma_1^E+\ldots+\gamma_W^E\stackrel{\text{a.s.}}{=}\lim_{N\to\infty}\frac{\log|f_N^E|}{N},
\end{equation*}
but \cref{eq:determinant-to-Lyapunov} is enough for our purposes. We are able to estimate the rate of 
convergence in \cref{eq:determinant-to-Lyapunov} (see \cref{prop:rate-of_convergence}) and reduce the problem
of finding a lower bound for the sum of Lyapunov exponents to finding a lower bound for $ \ex(\log|f_N^E|) $.
The idea behind estimating $ \ex(\log|f_N^E|) $ is very simple: if we have a random variable $ X\in[0,M] $ then
\begin{equation*}
	\ex X\ge M^{-1} \ex{X^2}\ge M^{-1} \var X.
\end{equation*}
Of course, $ \log|f_N^E| $ doesn't satisfy the needed assumptions, but we can argue that the values of 
$ \log|f_N^E| $ outside of $ [0,(NW)^{1/2+\epsilon}] $ do not have a significant contribution towards the 
expected value. For the values greater than $ (NW)^{1/2+\epsilon} $ this follows from the large deviations
estimate. For the negative values we use the following Cartan type estimate.
\begin{theorem}\label{thm:Cartan-W-negative-values}
	Let $ E\in \R $ and $ N,W\ge 1 $. There exists a constant $ C_0=C_0(D_0,D_1,|E|) $ such that 
	\begin{equation*}
		\pr(\log|f_N^E|< -10KW)\le\exp(-K/4).
	\end{equation*}
	for any $ K\ge C_0(1+\log(NW)) $.
\end{theorem} 
The above estimate is a crucial improvement over the straightforward estimate that can be obtained from the 
Wegner estimate: 
\begin{equation*}
	\pr(\log|f_N^E|<-KNW)\le\exp(-K/4).
\end{equation*}
for any $ K\ge C(1+\log(NW)) $ (see \cref{lem:Cartan-estimates} (a)).

Finally, we discuss the organization of the paper. \cref{sec:var,sec:ldt,sec:tm} deal, in order, with the 
proofs of \cref{thm:var-lower_bound,thm:ldt,thm:Cartan-W-negative-values}. These sections are independent of
each other, with the exception of a few auxiliary lemmas that get used throughout the paper.
The main result, \cref{thm:lower-bound-sum-Lyapunov}, is deduced in \cref{sec:lower-bound-expectation}. 
 
 % % % % % % % % % % % % % % % % % % % % % % % % % % % % % % 
\section{Lower Bound for the Variance}
\label{sec:var}
 % % % % % % % % % % % % % % % % % % % % % % % % % % % % % % 
 
\cref{thm:var-lower_bound} follows, with small modifications, from the proof of the lower bound for the 
fluctuations of Green's functions 
\cite[Thm. 1.1]{BGV-13-fluctuations}. For the convenience of the reader and in the interest of clarity we give
a complete proof in this section.

A key ingredient for the proof is the 
following estimate for the variance of a logarithmic potential. We use $ m_I $ to denote the uniform 
probability measure on a set $ I\subset \R $ and $ 
\var_I $ to 
denote the variance with respect to $ m_I $. We also use $ \norm{\cdot}_{I} $ to denote 
the norm in $ L^2(I,m_I) $.

\begin{proposition}\label{prop:estimate-for-log-potential}
	(\cite[Proposition 2.2 (iii)]{BGV-13-fluctuations})
  Let
  $\mu$
  be a Borel probability measure on
  $\R$
  and let
  \begin{equation*}
  u(x)=\int_{\R}\log |x-\zeta|d\mu(\zeta).
  \end{equation*}
  If
  $\mu ( |\zeta|\ge R)=0$
  for some
  $R>0$,
  then for any
  $M_1\ge 2M_0\ge 4R$
  one has
  \begin{equation*}
  \left|\var_{[M_0,M_1]}(u) -1\right|
  \lesssim (RM_1^{-1})^{1/5}+(M_0M_1^{-1})^{1/2}.
  \end{equation*}
\end{proposition}

We will also need the estimate on the integrability of the logarithmic potentials from
\cref{lem:log_potential-moments}.  Its proof uses the following standard lemma. We state it as a separate 
result because we also need it in the other sections. We use the notation $ \norm{X}_m $
for $ (\ex X^m)^{1/m} $.
\begin{lemma}\label{lem:abstract-moment-estimate}
	If $ X\ge 0 $ is a random variable such that
	\begin{equation*}
		\pr(X>C_0 K)\le \exp(-c_0K),
	\end{equation*}
	for every $ K\ge K_0 $, with $ c_0\le1 $, $ C_0,K_0\ge 1 $, then
	\begin{equation*}
		\norm{X}_m\lesssim m C_0 K_0/c_0,\,m\ge 1.
	\end{equation*}
\end{lemma} 

\begin{proof}
	\begin{multline*}
		\ex X^m=\int_{0}^{\infty}\pr(X> \lambda) m\lambda^{m-1}\,d\lambda
		=\int_{0}^{\infty}\pr(X>C_0 K) C_0^m m K^{m-1}\,dK\\
		\le (C_0K_0)^m+C_0^m\int_{K_0}^{\infty}\exp(-c_0K)m K^{m-1}\,dK
		\le (C_0K_0)^m+C_0^m \Gamma(m+1)/c_0^m. 
	\end{multline*}
	Applying Stirling's formula for the gamma function we conclude that
	\begin{equation*}
		\norm{X}_m\lesssim C_0K_0+mC_0/c_0.
	\end{equation*}
	Note that we stated the estimate in a weaker form because it is somewhat easier to apply and the weakening 
	doesn't affect the other estimates in this paper.
\end{proof}

\begin{lemma}\label{lem:log_potential-moments}
	Let $ \mu $ be a Borel measure on $ \R $ such that $ \mu(\R)\le 1 $ and
	\begin{equation*}
		\mu(|\zeta|>R)\le C_0/R,~R\ge C_1.	
	\end{equation*}
	Then for any non-degenerate interval $ I\subset\R $ we have
	\begin{equation*}
		\norm{\int_{\R}\log|x-\zeta|\,d\mu(\zeta)}_{I}
		\lesssim \max(1,\log M,-\log|I|,\log C_0,\log C_1),
	\end{equation*}
	where $ M=\sup_{I} |x| $.
\end{lemma}

\begin{proof}
	We have
	\begin{multline*}
		\mu\times m_I(\log|x-\zeta|<-K)=\int_\R \int_I \ind_{\{ |x-\zeta|<\exp(-K) 
		\}}\,dm_I(x)d\mu(\zeta)\\
		\le \int_\R \frac{2\exp(-K)}{|I|}\,d\mu(\zeta)\le\frac{2\exp(-K)}{|I|},
	\end{multline*}
	and
	\begin{multline*}
		\mu\times m_I(\log|x-\zeta|>K)
		\le \mu(|\zeta|>\exp(K)/2)+m_I(|x|>\exp(K)/2)\\
		=\mu(|\zeta|>\exp(K)/2)\le 2C_0 \exp(-K),
	\end{multline*}
	provided $ K\gg \max(1,\log M,\log C_1) $. It follows that
	\begin{equation*}
		\mu\times m_I(|\log|x-\zeta||>K)\le \exp(-K/2),
	\end{equation*}
	provided that $ K\gg \max(1,\log M,-\log|I|,\log C_0,\log C_1) $. Now the conclusion follows from 
	\cref{lem:abstract-moment-estimate} and the Cauchy-Schwarz inequality.
\end{proof}

For the convenience of the reader we state the basic general estimates on variance that we will be 
using.
\begin{lemma}\label{lem:var-elementary-estimates}
 Let
 $ (\Omega,\mc F,\mu) $
 be a probability space.
 \begin{enumerate}[(i)]
   \item
   If
   $ X $,
   $ Y $
   are square summable random variables then
   \begin{equation}
     |\var(X)-\var(Y)|\le \norm{X-Y}
       \left(\norm{X}+\norm{Y}\right),
   \end{equation}
	where $ \norm{\cdot} $ is the $ L^2 $ norm.
   \item
   If
   $ X $
   is a square summable random variable and
   $ \mc F_i $,
   $ i=1,\ldots,n $
   are pairwise independent
   $\sigma $-subalgebras of
   $ \mc F $
   then
   \begin{equation}\label{eq:var-Bessel}
   \var(X)\ge \sum_{i=1}^{n}\var(\ex(X\vert\mc F_i)).
   \end{equation}

   \item
   If
   $ X $
   is a square summable random variable and
   $ \mu_0 $
   is a probability measure  such that
   $ \mu \ge c\mu_0 $,
   with
   $ c\ge0 $,
   then
   \begin{equation}\label{eq:var-var>cvar}
   \var(X)\ge c \var_{\mu_0}(X).
   \end{equation}
  \end{enumerate}
\end{lemma}

\begin{proof}(of \cref{thm:var-lower_bound})
	By the Bessel type inequality \eqref{eq:var-Bessel} we get
	\begin{equation*}
		\var(\log|f_{\Lambda}^E| )\ge \sum_{k\in\Lambda} \var\left(\ex\left(\log|f_\Lambda^E|\big\vert V_k\right)\right).
	\end{equation*}
	We now just have to provide a lower bound for each term on the right-hand side of the above inequality. 
	We will achieve this by applying \cref{prop:estimate-for-log-potential}. 
	
	First we construct the logarithmic potential to which we will apply \cref{prop:estimate-for-log-potential}. 	
	We factorize $ f_{\Lambda}^E  $ by using Schur's formula (see for example \cite[Theorem 1.1]{Zha-05-Schur}). 
	In an appropriate basis we can write
	\begin{equation*}
		H_\Lambda-E=\begin{bmatrix}
			V_k-U_{k_1}(k_2,k_2)-E & \Gamma \\
			\Gamma^t & H_{\Lambda\setminus\{ k \}}-E
		\end{bmatrix},
	\end{equation*}
	where
	\begin{equation}\label{eq:gamma}
		\Gamma(k,j)=\begin{cases}
			-1& ,\text{ if } k_2=j_2 \text{ and } |k_1-j_1|=1\\
			-U_{k_1}(k_2,j_2)&, \text{ if } k_1=j_1 \text{ and } |k_2-j_2|\le d\\
			0&, \text{ otherwise}
		\end{cases}
	\end{equation}
	(the rows and columns are labeled by the indices of the potentials that they contain).
	By Schur's formula we have
	\begin{equation*}
		f_\Lambda^E=(V_k-\xi_k)
			\det(H_{\Lambda\setminus \{ k \}}-E),
	\end{equation*}
	where
	\begin{equation}\label{eq:xi_k}
		\xi_k=U_{k_1}(k_2,k_2)+E+\Gamma(H_{\Lambda\setminus \{ k \}}-E)^{-1}\Gamma^t.
	\end{equation}
	Since 
	$ \det(H_{\Lambda\setminus \{ k \}}-E) $ is independent of $ V_k $ it follows that
	\begin{equation*}
		\var\left(\ex\left(\log|f_\Lambda^E|\big\vert V_k\right)\right)
		=\var\left(\ex\left(\log|V_k-\xi_k|\big\vert V_k\right)\right).
	\end{equation*}
	and $ \var(h_k)=\var(u_k) $, with
	Let $ \mu_k $ be defined by
	$ \mu_k(S)=\pr(\xi_k\in S) $. Then we have
	\begin{equation*}
	 	u_k(x):=\ex\left(\log|V_k-\xi_k|\big\vert V_k\right)(x)=\int_{\R}\log|x-\zeta|\,d\mu_k(\zeta).
	\end{equation*} 
	Now that we have the logarithmic potential $ u_k $ we set things up for applying 
	\cref{prop:estimate-for-log-potential}. Let $ I=[A_0 R_0,A_1 R_0] $, with $ A_0,A_1,R_0>0 $ to be 
	chosen later. The proof is the 
	same for the case $ R_0<0 $ (corresponding to the case $ I=[-M_0,-C_1M_0] $ from the statement of 
	the theorem). By 
	\cref{lem:var-elementary-estimates} (iii) we have
	\begin{equation*}
		\var(u_k)\ge (\inf_I \rho)|I| \var_I(u_k).
	\end{equation*} 
	Let $ \mu_{k,1} $ and $ \mu_{k,2} $
	be defined by
	\begin{equation*}
		\mu_{k,1}(S)=\mu_k(S\cap[-R_0,R_0]),\qquad \mu_{k,2}(S)=\mu_k(S\setminus[-R_0,R_0]).
	\end{equation*}
	Let 
	$ u_{k,i}(x)=\int_{\R}\log|x-\zeta|\,d\mu_{k,i}(\zeta) $, $ i=1,2 $.
	By applying \cref{prop:estimate-for-log-potential} to $ u_{k,1}/\mu_k([-R_0,R_0]) $ we obtain
	\begin{equation*}
		\var_I(u_{k,1})
		\ge \frac{1}{2} (\mu_k([-R_0,R_0]))^{2},
	\end{equation*}
	provided $ 1\ll A_0\ll A_1 $.
	From \cite[Theorem II.1]{AM-93-Localization} we have
	\begin{equation}\label{eq:Wegner-for-entries}
		\pr(|(H_{\Lambda}-E)^{-1}(i,j)|\ge T)\lesssim D_0/T,
	\end{equation}
	for any $ \Lambda $ and any
	$ i,j\in \Lambda $. From this estimate and the integrability assumption \cref{eq:intro-densitydecay} 
	(see also \cref{eq:gamma} and \cref{eq:xi_k}) it follows that
	\begin{equation*}
		\mu_k(|\zeta|>R)=\pr(|\xi_k|>R)\le C d^2 (d^2/R)^{1/3},
	\end{equation*}
	for any $ R\gg d^2 |E| $ and with $ C=C(D_0,D_1) $ . As a consequence we get that
	\begin{equation*}
		(\mu_k([-R_0,R_0]))^{2}\ge\left( 1-C(d^8/R_0)^{1/3} \right)^2\ge 1/2,
	\end{equation*}
	provided $ R_0\ge C(D_0,D_1,|E|) d^8 $. So, if $ R_0 $ is large enough then 
	$ \var_I(u_{k,1})\ge 1/4 $. 
	
	The last step is to see that the bound on the fluctuations of $ u_{k,1} $
	implies a bound for the fluctuations of $ u_k $.
	By \cref{lem:var-elementary-estimates} (i) we have
	\begin{equation*}
		|\var_I(u_k)-\var_I(u_{k,1})|\le \norm{u_{k,2}}_{I}(\norm{u_{k,1}}_{I}+\norm{u_k}_{I}). 
	\end{equation*}
	From the Cauchy-Schwarz inequality we get
	\begin{equation*}
		\norm{u_{k,2}}_{I}\le \sqrt{\mu_k(|\zeta|>R_0)}\norm{u_k}_I\le C (d^8/R_0)^{1/6} \norm{u_k}_I.
	\end{equation*}	
	Now \cref{lem:log_potential-moments} implies that
	\begin{equation*}
		\norm{u_{k,2}}_{I}(\norm{u_{k,1}}_{I}+\norm{u_k}_{I})
		\le C(d^8/R_0)^{1/6}\log^2R_0\le 1/8,
	\end{equation*}
	provided $ R_0\ge C(D_0,D_1,|E|)d^9 $. Hence we have $ \var_I(u_k)\ge 1/8 $.
	
	We conclude that
	\begin{multline*}
		\var(\log|f_\Lambda^E|)\ge \sum_{k\in \Lambda} \var(h_k)\\=
		\sum_{k\in \Lambda} \var(u_k)\ge \sum_{k\in \Lambda} |I|(\inf_I \rho)\var_I(u_k)
		\ge |\Lambda||I|(\inf_I \rho)/8,
	\end{multline*}
	for any $ I=[A_0 R_0,A_1 R_0] $ with $ R_0\ge C(D_0,D_1,|E|)d^9 $ and $ A_1\gg A_0\gg 1 $. 
\end{proof}

 % % % % % % % % % % % % % % % % % % % % % % % % % % % % % % % 
\section{Large Deviations Estimate}\label{sec:ldt}
 % % % % % % % % % % % % % % % % % % % % % % % % % % % % % % % 

In this section we will prove \cref{thm:ldt}. The main idea is that $ \log|f_\Lambda^E| $ can be approximated 
 by a sum of independent random variables (alas, the error term is quite large). Namely, if $ \{ \Lambda_i \} $ is a partition of
$ \Lambda $  then we will see that
\begin{equation}\label{eq:almost-factorization}
	\log|f_\Lambda^E|\approx \sum_i \log|f_{\Lambda_i}^E|.
\end{equation}
The precise formulation is \cref{lem:Cartan-estimates} (b).
Once this is established we will obtain the large deviations estimates by applying the following exponential bound due to Bernstein (see \cite[Thm. 2.8]{Pet-95-Limit}).

\begin{theorem}[Bernstein]\label{thm:Bernstein}
	Let $ X_i $ be independent random variables such that $ \ex X_i=0 $, $ 
	i=1,\ldots,n $. Suppose that there exist positive constants $\sigma$ and $ T $ such that
	\begin{equation*}
		|\ex X_i^m|\le \frac{1}{2}m!\sigma^2 T^{m-2},~i=1,\ldots,n
	\end{equation*}
	for all integers $ m\ge 2 $. Then
	\begin{equation*}
		\pr \left( \left| \sum_{i=1}^{n}X_i \right|\ge x \right)\le \exp(-x/4T),\text{ if } x\ge 
		n\sigma^2/T.
	\end{equation*}
\end{theorem}

So, our first goal is to obtain \cref{eq:almost-factorization}. This will be a consequence of the following
general result. We use the notation $ \log^\pm:=\max(\pm\log,0) $.

\begin{lemma}\label{lem:interlacing}
	Let $ H_1 $ and $ H_2 $ be two self-adjoint operators on the same finite dimensional vector space. Then
	for any $ E\in \R $ we have
	\begin{multline*}
		\log|\det(H_1-E)|-\log|\det(H_2-E)|\\
		\le 4 
		\rank(H_1-H_2)\max(\log^+(|E|+\norm{H_1}),\log^-\dist(E,\spec H_2)).
	\end{multline*}
\end{lemma}
\begin{proof}
	Let $ N $ be the dimension of the vector space and let $ E^i_j $, $ j=1,\ldots,N $ be the eigenvalues of 
	$ H_i $,
	arranged in increasing order. We will use $ r $ to denote the rank of $ H_1-H_2 $. It is known that we have 
	the following 	interlacing inequalities due to Weyl (see \cite[Thm. 4.3.6]{HJ-85-Matrix}):
	\begin{equation*}
		E_j^1\le E_{j+r}^2,~j=1,\ldots,N-r,
	\end{equation*} 
	\begin{equation*}
		E_{j-r}^2\le E_j^1,~j=r+1,\ldots,N.
	\end{equation*}
	Let
	\begin{equation*}
		G_-= \{ j:~E_j^1-E<0,\,j\ge r+1 \},\quad G_+= \{ j:~E_j^1-E\ge 0,\,j\le N-r \}.
	\end{equation*}
	Then from the interlacing inequalities it follows 
	that
	\begin{equation*}
		|E_j^1-E|\le|E^2_{j-r}-E|,\,j\in G_-,\qquad |E^1_j-E|\le|E^2_{j+r}-E|,\,j\in G_+.
	\end{equation*}
	Let
	\begin{equation*}
		B_1=[1,N]\setminus(G_-\cup G_+),\quad B_2=[1,N]\setminus\left((-r+G_-)\cup (r+G_+)\right). 
	\end{equation*}
	Then we have
	\begin{multline*}
		\sum \log|E^1_j-E|-\sum \log |E^2_j-E|\\
		=\sum_{j\in G_-}(\log|E^1_j-E|-\log|E^2_{j-r}-E|)+\sum_{j\in G_+}(\log|E^1_j-E|-\log|E^2_{j+r}-E|)\\
			+\sum_{j\in B_1} \log |E^1_j-E|-\sum_{j\in B_2}\log|E^2_j-E|\\ 
		\le \sum_{j\in B_1} \log |E^1_j-E|-\sum_{j\in B_2}\log|E^2_j-E|\\
		\le 4r\max(\log^+(|E|+\norm{H_1}),\log^-\dist(E,\spec H_2)).
	\end{multline*}
	We used the fact that $ |B_1|,|B_2|\le 2r $. This concludes the proof.
\end{proof}
Of course, from the above lemma it follows that
\begin{multline*}
	|\log|\det(H_1-E)|-\log|\det(H_2-E)||\\
	\le 4 
	\rank(H_1-H_2) \max_i \max(\log^+(|E|+\norm{H_i}),\log^-\dist(E,\spec H_i)),
\end{multline*}
which is what interests us at the moment, but the one sided estimate stated in the lemma will be needed later
in \cref{lem:Dirichlet-vs-minors}.

We are ready to make \cref{eq:almost-factorization} precise. We also prove a related estimate needed for obtaining bounds on the moments of $ \log|f_\Lambda^E| $. Given two sets $ \Lambda_0\subset\Lambda\subset \Z_W $ we use $ \partial_\Lambda \Lambda_0 $ to denote the set
of $ i\in \Lambda\setminus \Lambda_0  $ such that there exists $ j\in\Lambda_0 $ that has a bond to $ i $.

% >LEM: CARTAN ESTIMATES
\begin{lemma}\label{lem:Cartan-estimates}
	Let $ \Lambda\subset \Z_W $ and $ E\in \R $. There exists $ C_0=C_0(D_0,D_1,|E|) $ such that the following 	statements are true for all $ K\ge C_0(1+\log|\Lambda|) $.
	\begin{enumerate}[(a)]
		\item We have
		 \begin{equation*}
			\pr(|\log|f_\Lambda^E||> |\Lambda|K)\le\exp(-K/4).
		\end{equation*}
		\item If $ \{ \Lambda_i \} $ is a partition of $ \Lambda $ then
		\begin{equation*}
			\pr\left(\left|\log|f_\Lambda^E|-\sum_i \log|f_{\Lambda_i}^E|\right|
				> 4|\cup_i \partial_\Lambda \Lambda_i |K\right)\le \exp(-K/4).
		\end{equation*}	
	\end{enumerate}
\end{lemma}

\begin{proof}
	(a) We have
		\begin{equation*}
			|\log|f_\Lambda^E||\le |\Lambda| \max \left(\log^+(|E|+\norm{H_\Lambda}),\log^
							-\dist(E,\spec H_\Lambda)  \right).
		\end{equation*}
	It follows that if $ 
	\log|f_{\Lambda}^E|>K|\Lambda|  $, then either $ \log(|E|+\norm{H_\Lambda})>K $ or 
	\begin{equation*}
		\log\dist(E,\spec H_\Lambda)<-K. 
	\end{equation*}
	Since $ \norm{H_\Lambda}\le 2+\max_{i} |V_i|+\max_n \norm{U_n} $, it 
	follows from \cref{eq:intro-densitydecay} that
	\begin{multline*}
		\pr(\log(|E|+\norm{H_\Lambda})>K)\\ 
		\le \pr(\log(|E|+2+\max_{i} |V_i|+\max_n \norm{U_n})>K)\\
		\le |\Lambda|\left[\pr(|V_i|>\exp(K)/3)+\pr(\norm{U_n}\ge \exp(K)/3)\right]\\
		\lesssim D_1 |\Lambda|\exp(-K)\le \exp(-K/2),
	\end{multline*}
	provided $ K\ge C (1+\log|\Lambda|) $, with $ C=C(D_1,|E|) $.
	Wegner's estimate (see \cite[(2.4)]{CGK-09-Generalized}) implies that
	\begin{equation*}
		\pr(\log(\dist(E,\spec H_\Lambda))<-K)\lesssim D_0 |\Lambda| \exp(-K)\le \exp(-K/2),
	\end{equation*}
	provided $ K\ge C(1+\log|\Lambda|) $, with $ C=C(D_0) $. Now the desired estimate follows 
	immediately.
  
	\noindent (b) 	Let $ H_\Lambda'=\oplus_i H_{\Lambda_i} $. Since we have that 
	$ \rank(H_\Lambda-H_\Lambda')\le |\cup_i \partial_\Lambda \Lambda_i| $ (the vectors that vanish on 
	$ \cup_i \partial_\Lambda \Lambda_i $ are in the kernel of $ H_\Lambda-H_\Lambda' $) it follows from
	\cref{lem:interlacing} that	
	\begin{multline*}
		\left|\log|f_\Lambda^E|-\sum_i \log|f_{\Lambda_i}^E|\right|\\
		\le 4
		\left| \cup_i \partial_\Lambda \Lambda_i  \right| 
			\max \left( \log^+(|E|+\norm{H_\Lambda}), \log^-\dist(E, \cup_i\spec \Lambda_i \cup 
			\spec \Lambda)\right).
	\end{multline*}
	We used the fact that we obviously have
	\begin{equation*}
		\norm{H_\Lambda'}\le\max_i \norm{H_{\Lambda_i}}\le \norm{H_\Lambda}.
	\end{equation*} 
	The desired estimate follows analogously to the proof of (a).
\end{proof}
% <LEM: CARTAN ESTIMATES

We are now ready to apply Bernstein's exponential bound.
The estimate depends on the moment estimates for $ \log|f_\Lambda^E| $. It turns out that the large 
deviations estimate that we obtain implies an improvement of the moment estimates which in turn lead to a
better large deviations estimate. So we will prove \cref{thm:ldt} through a recursion. We use the next 
proposition to facilitate the recursion.
 
\begin{proposition}\label{prop:ldt-from-moments}
	Let $ E\in \R $ and suppose that there exist positive constants $ C_0 $ and $ \delta_0\le 1/2 $ such that
	\begin{equation}\label{eq:moments-hypothesis}
		\norm{\log|f_{\Lambda}^E|-\ex(\log|f_{\Lambda}^E|) }_m
		\le m C_0|\Lambda|^{1/2+\delta_0}(1+\log|\Lambda|),
	\end{equation}
	for any rectangular $ \Lambda\subset \Z_W $ and any integer $ m\ge 2 $. Then there exists a constant 
	$ C_1=C_1(d,D_0,D_1,|E|) $ such that for $ K\gg 1 $ and any rectangular $ \Lambda $ 
	we have
	\begin{equation*}
		\pr\left(|\log|f_{\Lambda}^E|-\ex(\log|f_{\Lambda}^E| ) |
			>C_0C_1|\Lambda|^{1/2+c_0\delta_0}(1+\log|\Lambda|)K\right)
		\le \exp(-K/2),
	\end{equation*}
	with $ c_0=1/(1+2\delta_0) $. 	
\end{proposition}

\begin{proof}
	Let $ l $ denote the integer part of $ |\Lambda|^{1/2-c_0\delta_0} $ and let $ \{ \Lambda_i \} $ be the 	partition of $ \Lambda $ by the cells of the lattice 
	$ (l\Z)\times(l \Z) $ centered at the lower left corner 
	of $\Lambda$. Note that we have 
	\begin{equation*}
		|\Lambda_i|\le l^2 \le |\Lambda|^{1-2c_0\delta_0},\qquad
		|\cup_i \partial_\Lambda \Lambda_i|\lesssim d|\Lambda|/l\lesssim d|\Lambda|^{1/2+c_0\delta_0}.
	\end{equation*}
	We will obtain the conclusion from the inequality
	\begin{multline}\label{eq:ldt-triangle}
		\left| \log|f_\Lambda^E|-\ex(\log|f_\Lambda^E|) \right|
		\le \left| \log|f_\Lambda^E|-\sum_i\log|f_{\Lambda_i}^E| \right|
			+ \ex \left(\left| \log|f_\Lambda^E|-\sum_i\log|f_{\Lambda_i}^E| \right|\right)\\
			+ \left| \sum_i\left(\log|f_{\Lambda_i}^E|-\ex(\log|f_{\Lambda_i}^E|)\right) \right|, 			
	\end{multline}
	by estimating, with high probability, each of the terms on the right hand side.
	
	Applying \cref{lem:Cartan-estimates} (b)  we obtain
	\begin{multline}\label{eq:ldt-factorization}
		\left| \log|f_\Lambda^E|-\sum_i\log|f_{\Lambda_i}^E| \right|
		\le 4 C|\cup_i\partial_\Lambda \Lambda_i|(1+\log|\Lambda|)K\\
		\lesssim C d |\Lambda|^{1/2+c_0\delta_0}(1+\log|\Lambda|)K,
	\end{multline}
	except for a set of measure smaller than
	\begin{equation*}
		\exp(-C(1+\log|\Lambda|)K/4)\le \exp(-K),
	\end{equation*}
	provided $ C=C(D_0,D_1,|E|)\ge 4 $.
	
	From \cref{lem:Cartan-estimates} (b) and \cref{lem:abstract-moment-estimate} we get
	\begin{multline}\label{eq:ldt-average-factorization}
		\ex \left(\left| \log|f_\Lambda^E|-\sum_i\log|f_{\Lambda_i}^E| \right|\right)
		\le C |\cup_i\partial_\Lambda \Lambda_i|(1+\log |\Lambda|)\\
		\lesssim C d|\Lambda|^{1/2+c_0\delta_0}(1+\log |\Lambda|),
	\end{multline}
	with $ C=C(D_0,D_1,|E|)$.
	
	To estimate the last term on the right-hand side of \cref{eq:ldt-triangle} we will use 
	\cref{thm:Bernstein}. For this we need to estimate the number of sets in the partition 
	$ \{ \Lambda_i \} $. Depending on the proportions of $ \Lambda $ the bound can range from 
	$ |\Lambda|/l^2 $ (the ``typical'' case) to $ |\Lambda|/l $ (when $ \Lambda $ is a very narrow strip).
	Each case can be dealt with similarly, but the choices of constants in \cref{thm:Bernstein} need
	to be adjusted. To account for these adjustments we separate the partition into sets of the same size.
	Let $ I_k $, $ 1\le k\le 4 $, denote the sets of indices corresponding to the 
	maximal subfamilies 
	of $ \{ \Lambda_i \} $ of sets with the same dimensions.
	\begin{center}
		\includegraphics{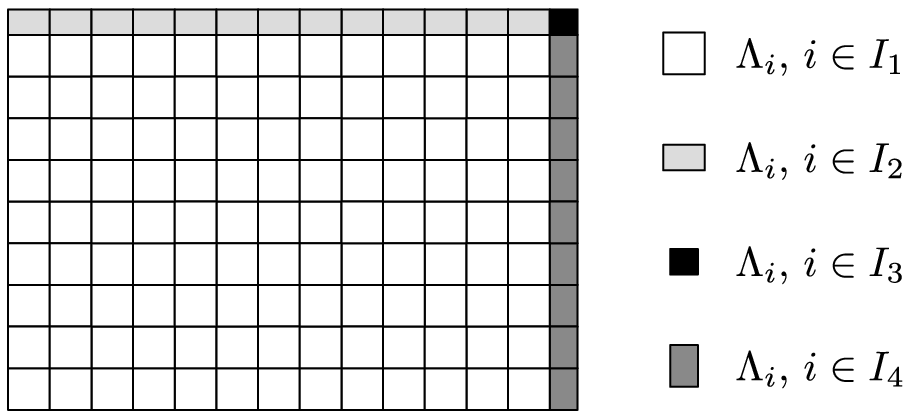}
	\end{center}
	Let $ A_k $ be the size of the sets $ \Lambda_i $ with indices in $ I_k $.
	We will apply \cref{thm:Bernstein} with 
	$ X_i=A_k^{-(1/2+\delta_0)}(\log|f_{\Lambda_i}^E|-\ex(\log|f_{\Lambda_i}^E|)) $, $ i\in 
	I_k $.
	From the hypothesis \cref{eq:moments-hypothesis} it follows that 
	\begin{equation*}
		|\ex X_i^m|
		\le \frac{1}{A_k^{m(1/2+\delta_0)}}
			\left( mC_0 A_k^{1/2+\delta_0}(1+\log A_k) \right)^m
		\le \frac{1}{2}m!\sigma_k^2T_k^{m-2},
	\end{equation*}
	with $ T_k\ge \sigma_k:=CC_0(1+\log A_k) $, $ C\gg 1 $.
	Applying \cref{thm:Bernstein}  we get 
	\begin{equation}\label{eq:Bernstein-k}
		\pr\left(\left| \sum_{i\in I_k}\left(\log|f_{\Lambda_i}^E|
			-\ex(\log|f_{\Lambda_i}^E|)\right) \right|\ge x_k A_k^{1/2+\delta_0}\right)
		\le \exp(-x_k/4T_k), 
	\end{equation}
	provided $ x_k\ge |I_k| \sigma_k^2/T_k $. In particular, a straightforward computation
	shows that \cref{eq:Bernstein-k} holds with $ x_k=4K\sigma_k\sqrt{|\Lambda|/A_k} $
	and 
	$ T_k=\sigma_k\sqrt{|\Lambda|/A_k} $, provided $ K\ge 1/4 $. Since 
	\begin{equation*}
		x_k A_k^{1/2+\delta_0}= 4K\sigma_k|\Lambda|^{1/2} A_k^{\delta_0}
		\le 4K\sigma_k|\Lambda|^{1/2}|\Lambda|^{\delta_0(1-2c_0\delta_0)}
		= 4K\sigma_k|\Lambda|^{1/2+c_0\delta_0},
	\end{equation*}
	it follows that for any $ k $ we have
	\begin{multline}\label{eq:ldt-I_k}
		\pr\left(\left| \sum_{i\in I_k}\left(\log|f_{\Lambda_i}^E|
			-\ex(\log|f_{\Lambda_i}^E|)\right) \right|
			\ge CC_0|\Lambda|^{1/2+c_0\delta_0}(1+\log|\Lambda|)K\right)\\
		\le \exp(-K). 		
	\end{multline}
	Plugging the estimates \eqref{eq:ldt-factorization}, \eqref{eq:ldt-average-factorization}, and
	\eqref{eq:ldt-I_k} into \eqref{eq:ldt-triangle} yields
	\begin{equation*}
		\left| \log|f_\Lambda|-\ex(\log|f_\Lambda|) \right|
		\le CC_0|\Lambda|^{1/2+c_0\delta_0}(1+\log|\Lambda|)K,
	\end{equation*}
	except for a set of measure less than $ 5\exp(-K)\le \exp(-K/2) $, provided $ K\gg 1 $.
	This concludes the proof.
\end{proof}

\begin{proof}(of \cref{thm:ldt})
	From \cref{lem:Cartan-estimates} (a) and \cref{lem:abstract-moment-estimate} it follows that for any
	$ m\ge 1 $ and any $ \Lambda\subset \Z_W $ we have
	\begin{equation*}
		\norm{\log|f_{\Lambda}^E|-\ex(\log|f_{\Lambda}^E|) }_m
		\le m C|\Lambda|^{1/2+\delta_0}(1+\log|\Lambda|),
	\end{equation*}
	with $ C=C(D_0,D_1,|E|) $ and $ \delta_0=1/2 $.	By applying \cref{prop:ldt-from-moments} and 
	\cref{lem:abstract-moment-estimate} $ n $ times we obtain
	\begin{equation*}
		\pr\left(|\log|f_{\Lambda}^E|-\ex(\log|f_{\Lambda}^E| ) |
			>C (C')^n|\Lambda|^{1/2+\delta_n}(1+\log|\Lambda|)K\right)
		\le \exp(-K/2),
	\end{equation*}	
	with $ K\gg 1 $, $ C'=C'(d,D_0,D_1,|E|) $, and $ \delta_n $ defined recursively by 
	\begin{equation*}
		\delta_n=\delta_{n-1}/(1+2\delta_{n-1}).	
	\end{equation*}
	The conclusion follows immediately by noticing that $ \delta_n\to 0 $ as $ n\to \infty $ (in fact we have
	$ \delta_n=1/(2n+2) $).
\end{proof}
% SEC: TRANSFER MATRICES AND DIRICHLET DETERMINANTS
\section{Transfer Matrices and Dirichlet Determinants}
\label{sec:tm}

In this section we prove \cref{thm:Cartan-W-negative-values}.  
 The idea for the proof of 
\cref{thm:Cartan-W-negative-values} is to exploit the fact that $ \det T_N^E=1 $. In the case $ W=1 $ it is well-known that 
\begin{equation*}
	T_N^E= \begin{bmatrix}
		f^E_{[1,N]} & -f^E_{[2,N]}\\[0.5em]
		f^E_{[1,N-1]} & -f^E_{[2,N-1]}
	\end{bmatrix}.
\end{equation*}
So it follows that 
\begin{equation*}
	1\le \norm{T_N^E}\le |f^E_{[1,N]}|+|f^E_{[1,N-1]}|+|f^E_{[2,N-1]}|+|f^E_{[2,N-1]}|.
\end{equation*}
This implies that at least one of the determinants cannot be smaller than $ 1/4 $. The conclusion of 
\cref{thm:Cartan-W-negative-values} would follow by noticing that all the determinants are roughly the same due
to \cref{lem:interlacing}.  To apply this strategy for the general case $ W\ge 1 $ we will work with the $ W $-th exterior
power of $ T_N^E $. We refer to \cite[III.5]{BL-85-Products} for a brief review of the exterior algebra.
The entries of the matrix $ \Exterior^W T_N^E $ are the $ W\times W $ minors of $ T_N^E $. Let us 
be more explicit. We use $ e_i $, $ i=1,\ldots,2W $ to denote the standard basis of $ \R^{2W} $. For any $ 
\alpha\subset \{ 1,\ldots,2W \} $ let $ e_\alpha= \Exterior_{i\in\alpha}e_i$.  The vectors $ e_\alpha 
$, with $ |\alpha|=W $ form the standard basis of $ \Exterior^W \R^{2W} $. For any decomposable vector $ 
u=u_1\wedge\ldots\wedge u_W $ we will use $ [u] $ to denote the matrix with columns $ u_1,\ldots,u_W $.
The 
space $ \Exterior^W \R^{2W} $ is equipped with an 
inner product defined through the formula
\begin{equation*}
	\langle u,v\rangle=\det([u]^t [v]), 
\end{equation*}
where $ u $ and $ v $ are decomposable vectors and $ [u]^t $ denotes the transpose of $ [u] $. 
The entries of $ \Exterior^W T_N^E $ (with respect to the standard basis) are of the form
\begin{equation*}
	\langle e_\beta, \Exterior^W T_N^E e_\alpha \rangle
	=\det([e_\beta]^t T_N^E [e_\alpha]).
\end{equation*}
For our purposes we need to figure out what is the connection between the entries of $ \Exterior^W T_N^E $
and Dirichlet determinants. It is not hard to see that $ \det([e_\beta]^t T_N^E [e_\alpha])=f_N^E $, when 
$ \alpha=\beta= \{ 1,\ldots,W \} $. For example, this is a consequence of the following lemma. This is also
known from \cite[Prop. 3.1]{CS-83-Log}, but the next result is crucial for us because it holds even for 
non-symmetric matrices.
\begin{lemma}\label{lem:tm-abstract-det-tm-det-Dir}
	If $ M_k $, $ k=1,\ldots,N $ are $ W\times W $ matrices then the 
	determinant of
	\begin{equation*}
		\begin{bmatrix}
				M_1 & - I &  &  & & & \\
				-I & M_2 & -I &  & & &\\
	 		 & \ddots & \ddots & \ddots & & &\\
				& & & \ddots & \ddots & \ddots & \\
				& & &  & -I & M_{N-1} & -I\\
				& & &  &  & -I & M_N\\
		\end{bmatrix}
	\end{equation*}
	is equal to the determinant of 
	\begin{equation*}
		\begin{bmatrix}
			I & 0
		\end{bmatrix}
		\left(
		\prod_{k=N}^{1}\begin{bmatrix}
			M_k & -I \\ I & 0
		\end{bmatrix}\right) 
		\begin{bmatrix}
			I \\ 0
		\end{bmatrix},
	\end{equation*}
	or, in other words, the determinant of the  top-left $ W\times W $ block of the transfer matrix.
\end{lemma}
\begin{proof}
	By continuity it is enough to prove the result for the case when the matrices $ M_k $ are invertible. The 	proof is by induction on $ N $. The case $ N=1 $ is trivially true. We assume the statement to be true for 
	$ N $ and we prove it for $ N+1 $. Let $ A_N $ denote the $ (NW)\times (NW) $ matrix from the statement 
	of the lemma. We can
	write 
	\begin{equation*}
		A_{N+1}= \begin{bmatrix}
			A_N & \Gamma\\
			\Gamma^t & M_{N+1}
		\end{bmatrix}
	\end{equation*}  
	with $ \Gamma $ a $ (N-1)W\times W $ matrix with the bottom $ W\times W $ block equal to $ -I $ and all the 	other entries equal to zero. By Schur's formula we have
	\begin{equation}\label{eq:Schur}
		\det A_{N+1}=\det M_{N+1} \det(A_N-\Gamma M_{N+1}^{-1}\Gamma^t). 
	\end{equation} 
	A direct computation shows that
	\begin{equation*}
		A_N-\Gamma M_{N+1}^{-1}\Gamma^t
		=\begin{bmatrix}
				M_1 & - I &  &  & & & \\
				-I & M_2 & -I &  & & &\\
	 		 	& \ddots & \ddots & \ddots & & &\\
				& & & \ddots & \ddots & \ddots & \\
				& & &  & -I & M_{N-1} & -I\\
				& & &  &  & -I & M_N-M_{N+1}^{-1}\\
		\end{bmatrix}.
	\end{equation*}
	By the induction hypothesis we have
	\begin{multline*}
		\det(A_N-\Gamma M_{N+1}^{-1}\Gamma^t)\\
		=\det \left( \begin{bmatrix}
			I & 0
		\end{bmatrix}
		\begin{bmatrix}
			M_N-M_{N+1}^{-1} & -I \\
			I & 0
		\end{bmatrix}
		\left(\prod_{k=N-1}^{1} \begin{bmatrix}
			M_k & -I \\
			I & 0
		\end{bmatrix}\right)
		\begin{bmatrix}
			I \\ 0
		\end{bmatrix}
		\right).
	\end{multline*}
	The conclusion follows from the above and \cref{eq:Schur} by noticing that
	\begin{equation*}
		M_{N+1}\begin{bmatrix}
			I & 0
		\end{bmatrix}
		\begin{bmatrix}
			M_N-M_{N+1}^{-1} & -I \\
			I & 0
		\end{bmatrix}
		=
		\begin{bmatrix}
			I & 0
		\end{bmatrix}
		\begin{bmatrix}
			M_{N+1} & -I\\
			I & 0
		\end{bmatrix}
		\begin{bmatrix}
			M_{N} & -I\\
			I & 0
		\end{bmatrix}.
	\end{equation*}
\end{proof} 

Not all of the entries of $ \Exterior^W T_N^E $ are determinants of self-adjoint matrices (see 
\cite[Prop. 3.1]{CS-83-Log} and the ensuing comments). Furthermore, it is not clear wether all the entries
are related to eigenvalue problems. To better understand this let us
discuss the eigenvalue problems associated with $ \det([v]^t T_N^E [u]) $, where $ u,v $ are non-trivial 
decomposable vectors in $ \Exterior^W \R^{2W} $ (of course, our immediate interest is in the case when $ u,v $ are vectors from the standard basis). 
We follow \cite[III.5.1]{CL-90-Spectral}. Let $ u,v $ be two non-trivial decomposable vectors in 
$ \Exterior^W \R^{2W} $. We have that $ \det([v]^t T_N^E [u])=0 $ if and only if there exists a non-trivial 
vector $ \Phi\in \R^{W} $ such that $ [v]^t T_N^E [u]\Phi =0 $.
Starting with $ \Phi $ we can use the transfer matrix to build a solution $ \Psi\in l^2([0,N+1],\C^W) $ of
\begin{equation*}
	(H\Psi)_i=E\Psi_i,\,i\in[1,N].
\end{equation*}
The solution $ \Psi $ is defined by
\begin{equation*}
	\begin{bmatrix}
		\Psi_{k+1} \\ \Psi_k
	\end{bmatrix}=T_k^E [u]\Phi,\quad k=0,\ldots,N,
\end{equation*}
where we let $ T_0^E $ be the identity matrix. Clearly the solution satisfies the boundary conditions
\begin{equation}\label{eq:boundary-conditions}
	\begin{bmatrix}
		\Psi_1 \\ \Psi_0
	\end{bmatrix} \in \Ran([u]),\qquad
	\begin{bmatrix}
		\Psi_{N+1} \\ \Psi_N
	\end{bmatrix} \in \Ker([v]^t).
\end{equation}
Given a decomposable vector $ w $ we will  use $ A_w $ and $ B_w $ to denote the top and bottom $ W\times W $ blocks of $ [w] $, so we have
\begin{equation*}
	[w]= \begin{bmatrix}
		A_w\\ B_w
	\end{bmatrix}.
\end{equation*}
Assuming that $ A_u $ and $ A_v $ are invertible it follows that the boundary conditions \cref{eq:boundary-conditions} are equivalent to
\begin{equation*}
	\Psi_0=B_uA_u^{-1}\Psi_1,\qquad \Psi_{N+1}=-(B_vA_v^{-1})^t\Psi_N.
\end{equation*}
Based on this we define the following operator on $ l^2([1,N],\C^W) $:
\begin{equation*}
	(H_N(u,v)\Psi)_i
	=
	\begin{cases}
		-\Psi_2+(S_1-B_uA_u^{-1})\Psi_1 & ,i=1\\
		(H\Psi)_i & ,i\in[2,N-1]\\
		-\Psi_{N-1}+(S_N+(B_vA_v^{-1})^t)\Psi_N & ,i=N		
	\end{cases}.
\end{equation*}
Let $ f_N^E(u,v)=\det(H_N(u,v)-E) $. From the construction of $ H_N(u,v) $ it follows that $ \det([v]^t T_N^E [u])=0 $ if and only if $ E $ is an eigenvalue for $ H_N(u,v) $. Hence, it is not surprising that $ \det([v]^t T_N^E [u]) $ and $ f_N^E(u,v) $ are the same up to a multiplicative constant. 
More precisely we have the following result. 
\begin{proposition}
	Let $ E\in \C $ and let $ u $ and $ v $ be decomposable vectors such that $ A_u $ and $ A_v $ are 
	invertible. Then we have 
	\begin{equation*}
		\det(A_uA_v) f_N^E(u,v)=\det([v]^t T_N^E [u]).
	\end{equation*}
\end{proposition}
\begin{proof}
	Note that
	\begin{multline*}
		[v]^t T_N^E [u]
		\\=
		A_v^t \begin{bmatrix}
			I & (B_vA_v^{-1})^t
		\end{bmatrix}
		\begin{bmatrix}
			S_N-E & -I \\ I & 0
		\end{bmatrix}  
		\left(\prod_{k=N-1}^{2}\begin{bmatrix}
			S_k-E & -I\\
			I & 0
		\end{bmatrix}\right)\\
		\cdot
		\begin{bmatrix}
			S_1-E & -I \\ I & 0
		\end{bmatrix}
		\begin{bmatrix}
			I \\ B_uA_u^{-1}
		\end{bmatrix} A_u\\
		\\=
		A_v^t
		\begin{bmatrix}
			I & 0
		\end{bmatrix} 
		\begin{bmatrix}
			S_N+(B_vA_v^{-1})^t-E & -I\\
			I & 0
		\end{bmatrix}
		\left(\prod_{k=N-1}^{2}\begin{bmatrix}
			S_k-E & -I\\
			I & 0
		\end{bmatrix}\right)\\
		\cdot\begin{bmatrix}
			S_1-B_uA_u^{-1}-E & -I\\
			I & 0
		\end{bmatrix}
		\begin{bmatrix}
			I \\ 0
		\end{bmatrix} 
		A_u.
	\end{multline*}
	The conclusion now follows from \cref{lem:tm-abstract-det-tm-det-Dir}.
\end{proof}

When $ u,v $ are part of the standard basis we don't have, in general, that $ A_u,A_v $ are invertible. So it is not clear what eigenvalue problems can be associated with $ \det([e_\beta]^tT_N^E [e_\alpha]) $. To make use of the above discussion we will work with a different basis of $ \Exterior^W\R^{2W} $. Clearly, we need to have control 
on the norms of $ A_u^{-1} $ and $ B_u $ for all $ u $ in the new basis. We deal with these issues in the 
following.
\begin{lemma}\label{lem:basis}
	There exists a basis $ \{ u_\alpha \} $ of $ \Exterior^W\R^{2W} $ with $ A_{u_\alpha}=I $ and 
	$ \norm{B_{u_\alpha}}\le 1 $ for all $ \alpha $, such that the vectors from the standard basis can be 
	written as
	linear combinations of $ \{ u_\alpha \} $  with coefficients having absolute value lesser or equal to 
	$ 1 $.
\end{lemma}
\begin{proof}
	We build the basis explicitly. For any $ \alpha\subset [1,2W] $, with $ |\alpha|=W $, we let 
	$ u_\alpha=\Exterior_{i\in[1,W]}u_{\alpha,i} $, where
	\begin{equation*}
		u_{\alpha,i}=\begin{cases}
			e_i &, i\in\alpha\\
			e_i+e_{\phi_\alpha(i)} &, i\notin \alpha
		\end{cases}
	\end{equation*}
	and $ \phi_\alpha $ is a bijection from $ [1,W]\setminus \alpha $ to $ \alpha\cap [W+1,2W] $.
	We clearly have that $ A_{u_\alpha}=I $ and $ \norm{B_{u_\alpha}}\le 1 $ ($ B_{u_\alpha} $ is the matrix 
	with the $ i $-th column equal to zero if $ i\in\alpha $ and equal to $ d_{\phi_\alpha(i)-W} $ if 
	$ i\notin \alpha $, where $ \{ d_i \}$ is the standard basis of $ \R^W $). 
	
	Now we just have to check that $ \{ u_\alpha \} $ generates the standard basis, with the stated bounds on 
	the coefficients. Fix $ \alpha\subset [1,2W] $, with $ |\alpha|=W $. Let the elements of $ \alpha $ be 
	$ \alpha_1\le \ldots\le \alpha_W $ and let $ k $ be such that 
	$ \alpha\cap[1,W]= \{ \alpha_1,\ldots,\alpha_k \} $. The conclusion follows by writing 
	\begin{multline*}
		e_\alpha=e_{\alpha_1}\wedge\ldots\wedge e_{\alpha_k}\\
			\wedge \left[
				\left(e_{\alpha_{k+1}}+e_{\phi_\alpha^{-1}(\alpha_{k+1})}\right)
					-e_{\phi_\alpha^{-1}(\alpha_{k+1}
			)}\right]\wedge \ldots
			\wedge \left[
				\left(e_{\alpha_{W}}+e_{\phi_\alpha^{-1}(\alpha_{W})}\right)
					-e_{\phi_\alpha^{-1}(\alpha_{W}
			)}\right] 
	\end{multline*}	  
	and expanding the square brackets.
\end{proof}

From the previous lemma it follows that for any $ \alpha_0,\beta_0 $ we have
\begin{equation*}
	|\langle e_{\beta_0},\Exterior^W T_N^E e_{\alpha_0} \rangle|
	\le \sum |\langle u_\alpha,\Exterior^W T_N^E u_\beta \rangle |
\end{equation*}
and consequently
\begin{equation}\label{eq:norm-minors-bound}
	\norm{\Exterior^W T_N^E}\le \exp(CW) \sum |\det([u_\beta]^t T_N^E [u_\alpha])|,
\end{equation}
with $ C $ an absolute constant. We would now like to argue that the determinants on the right-hand side are 
roughly the same as $ f_N^E $. This only works for the determinants corresponding to symmetric matrices. 
However, we will only need the following weaker estimate that also holds for the determinants of 
non-symmetric matrices.

\begin{lemma}\label{lem:Dirichlet-vs-minors}
	Let $ E\in \R $ and let $ \{ u_\alpha \} $ be the basis from \cref{lem:basis}. Then there exists a constant 
	$ C_0=C_0(D_0,D_1,|E|) $ such that
	\begin{equation*}
		\pr \left(\exists\, \alpha, \beta,\, \log|f_N^E(u_\alpha,u_\beta)|-\log|f_N^E|>8KW \right)
		\le \exp(-K/4),
	\end{equation*}
	for any $ K\ge C_0(1+\log(NW)) $.
\end{lemma}
\begin{proof}
	Let 
	\begin{equation*}
		\tilde H_1=(H_N(u_\alpha,u_\beta)-E)^t(H_N(u_\alpha,u_\beta)-E)	
	\end{equation*}
	and
	\begin{equation*}
		\tilde H_2=(H_N-E)^t(H_N-E).
	\end{equation*}
	By a direct computation of $ \tilde H_1-\tilde H_2 $ one sees that $ \rank(\tilde H_1-\tilde H_2)\le 4W $.
	Applying \cref{lem:interlacing} we obtain
	\begin{multline*}
		2\log|f_N^E(u_\alpha,u_\beta)|-2\log|f_N^E|=\log|\det \tilde H_1|-\log|\det \tilde H_2|\\
		\le 16W \max(\log^+ \norm{\tilde H_1},\log^-\dist(0,\spec \tilde H_2)).
	\end{multline*}
	Next we note that
	\begin{multline*}
		\norm{\tilde H_1}\le \norm{H_N(u_\alpha,u_\beta)-E}^2\\
		\le \left(2+\max_i |V_i|+\max_n \norm{U_n}+|E|+
			\norm{B_{u_\alpha} A_{u_\alpha}^{-1}}+\norm{B_{u_\beta} A_{u_\beta}^{-1}}\right)^2\\
		\le (\max_i |V_i|+\max_n \norm{U_n}+|E|+4)^2
	\end{multline*}
	and
	\begin{equation*}
		\dist(0,\spec \tilde H_2)=(\dist(E,\spec H_N))^2.
	\end{equation*}
	From the above relations we can now conclude (as in the proof of \cref{lem:Cartan-estimates}) that
	\begin{multline*}
		\pr \left(\exists\, \alpha, \beta,\, \log|f_N^E(u_\alpha,u_\beta)|-\log|f_N^E|>8KW \right)\\
		\le \pr(\log \norm{\tilde H_1}>K)+\pr(\log(\dist(0,\spec \tilde H_2))<-K)\\
		\le \pr(\max_i |V_i|>\exp(-K/2)/3+\pr(\max_n \norm{U_n}>\exp(-K/2)/3)\\
		+\pr(\dist(E,\spec H_N)<\exp(-K/2))
		\le \exp(-K/4),
	\end{multline*}
	provided $ K\ge C(1+\log(NW)) $, with $ C=C(D_0,D_1,|E|) $.
	
	Note that it is not possible to repeat the argument by switching $ \tilde H_1 $ with $\tilde  H_2 $ because
	we don't have a Wegner estimate for the non-symmetric matrix $ \tilde H_2 $.
\end{proof}

We have all we need to prove \cref{thm:Cartan-W-negative-values}. 
\begin{proof}(of \cref{thm:Cartan-W-negative-values})
	We have that $ \norm{\Exterior^W T_N^E}^{\binom{2W}{W}}\ge |\det \Exterior^W T_N^E|=1 $. The last identity 
	follows from the Sylvester-Franke Theorem (see, for example, \cite{Tor-52-Sylvester}). Using 
	\cref{eq:norm-minors-bound} it follows that $ \max_{\alpha,\beta} |f_N^E(u_\alpha,u_\beta)| \ge 
	\exp(-CW) $, with some absolute constant $ C $. From \cref{lem:Dirichlet-vs-minors} we can conclude that
	\begin{equation*}
		\log|f_N^E|\ge \max_{\alpha,\beta} \log|f_N^E(u_\alpha,u_\beta)|-8KW\ge -10KW,
	\end{equation*}
	except for a set of measure less than $ \exp(-K/4) $, for any $ K\ge C(1+\log(NW)) $ with 
	$ C=C(D_0,D_1,|E|) $ large enough.
\end{proof}
% <PROP: BETTER CARTAN FOR NEGATIVE VALUES OF log|f_\Lambda|

% >SEC: LOWER BOUND FOR THE EXPECTED VALUE
\section{Proof of the Main Result}\label{sec:lower-bound-expectation}

In this section we prove \cref{thm:lower-bound-sum-Lyapunov}. As was mentioned in the introduction, we will
first estimate the rate of convergence in \cref{eq:determinant-to-Lyapunov} and provide a lower bound for
$ \ex(\log|f_N^E|) $.

% >LEM: MULTISCALE COMPARISON FOR \ex\log|f_N|/N
We will use the following lemma to estimate the rate of convergence in \cref{eq:determinant-to-Lyapunov}.
\begin{lemma}\label{lem:multiscale-comparison-ex-f_N}
	Let $ E\in \R $. There exists a constant $ C_0=C_0(D_0,D_1,|E|) $ such that
	\begin{equation*}
		\left| \frac{\ex(\log|f_{N_1}^E|)}{N_1}-\frac{\ex(\log|f_{N_2}^E|)}{N_2} \right|
		\le C_0\frac{W\log\left(N_1 W \right)}{N_2},
	\end{equation*}
	for any $ N_1\ge N_2^2> 1 $.
\end{lemma}
\begin{proof}
	Let $ n=[N_1/N_2] $ and
	\begin{equation*}
		\Lambda_i=\left( [1,N_1]\cap [iN_2+1,(i+1)N_2] \right)\times [1,W], i=0,\ldots,n.
	\end{equation*}
	Applying  \cref{lem:Cartan-estimates} (b) and \cref{lem:abstract-moment-estimate} we obtain
	\begin{multline*}
		\left|\ex(\log|f_{N_1}^E|)-\sum_{i=1}^{n} \ex \left( \log|f_{\Lambda_i}^E| \right)\right|
		=\left|\ex(\log|f_{N_1}^E|)-n \ex \left( \log|f_{N_2}^E| \right)
			-\ex \left( \log|f_{\Lambda_n}^E| \right)\right|\\
		\le CnW\log(N_1W).
	\end{multline*} 
	It follows that
	\begin{equation*}
		\left| \frac{\ex(\log|f_{N_1}^E|)}{N_1}-\frac{\ex(\log|f_{N_2}^E|)}{N_2} \right|
		\le C\frac{W\log\left( N_1 W \right)}{N_2}+\frac{\ex(|\log|f_{N_2}^E||)}{N_1}
			+\frac{\ex(|\log|f_{\Lambda_{n}}^E||)}{N_1}.		
	\end{equation*}
	By \cref{lem:Cartan-estimates} (a), \cref{lem:abstract-moment-estimate}, and the assumption that 
	$ N_1\ge N_2^2 $ we have
	\begin{multline*}
		\frac{\ex(|\log|f_{N_2}^E||)}{N_1}+\frac{\ex(|\log|f_{\Lambda_{n}}^E||)}{N_1}
		\le \frac{CN_2W\log(N_2W)}{N_1}+\frac{CN_2W\log(N_2W)}{N_1}\\
		\lesssim C\frac{W\log(N_1W)}{N_2}.
	\end{multline*}
	The conclusion follows immediately.
\end{proof} 
% <LEM: MULTISCALE COMPARISON FOR \ex\log|f_N|/N

% >PROP: RATE OF CONVERGENCE OF log(|f_N|)/N
\begin{proposition}\label{prop:rate-of_convergence}
	Let $ E\in \R $. There exists a constant $ C_0=C_0(D_0,D_1,|E|) $ such that for any $ N>1 $ we have
	\begin{equation*}
		\left|\gamma_1^E+\ldots+\gamma_W^E-\frac{\ex(\log|f_N^E|)}{N}\right|\le C_0\frac{W\log(NW)}{N}.
	\end{equation*}
\end{proposition}
\begin{proof}
	Let $ N_k=N^{2^k} $. By applying \cref{lem:multiscale-comparison-ex-f_N} $ k $ times we get
	\begin{equation}\label{eq:multiscale-N_k}
		\left|\frac{\ex(\log|f_{N_k}^E|)}{N_k}-\frac{\ex(\log|f_N^E|)}{N}\right|
		\le \sum_{j=0}^{k-1} \frac{CW\log(N_{j+1}W)}{N_j}.
	\end{equation}	
	Since
	\begin{multline*}
		\sum_{j=0}^{\infty} \frac{CW\log(N_{j+1}W)}{N_j}
		=CW\left(\frac{\log N}{N}\sum_{j=0}^{\infty}\frac{2^{j+1}}{N^{2^j-1}}
			+ \frac{\log W}{N}\sum_{j=0}^{\infty}\frac{1}{N^{2^j-1}}\right)\\
		\lesssim \frac{CW\log(NW)}{N},
	\end{multline*}
	the conclusion follows from \cref{eq:determinant-to-Lyapunov} by letting $ k\to \infty $ in 
	\eqref{eq:multiscale-N_k}.
\end{proof}
% <PROP: RATE OF CONVERGENCE OF log(|f_N|)/N

% >LEM: LOWER BOUND FOR log|f_N|
\begin{lemma}\label{lem:lower-bound-determinant}
	Assume that $ \var(\log|f_N^E|)\ge \beta_0 NW $. Let $ 0<\epsilon\ll 1 $. Then there exists  
	a constant $ C_0=C_0(D_0,D_1,|E|,\beta_0,\epsilon) $ such that we have
	\begin{equation*}
		\ex(\log|f_{N}^E| )\ge \frac{\beta_0}{8}(NW)^{1/2-\epsilon},
	\end{equation*}
	for any 
	$ N\ge C_0 W^{1+5\epsilon}$.
\end{lemma}

\begin{proof}
	We can assume that $ \ex(\log|f_N^E| )\le (NW)^{1/2} $ because otherwise there is nothing to prove. Let 
	$ \Omega_-=\{ \log|f_{N}^E|<0  \} $, $ \Omega=\{ 
	\log|f_{N}^E|\in[0,2(NW)^{1/2+\epsilon}]  \} $, and
	$ \Omega_+=\{ \log|f_{N}^E| > 2(NW)^{1/2+\epsilon}\} $. Also let 
	\begin{equation*}
		X_-=\ind_{\Omega_-}\log|f_N^E|,\quad
		X=\ind_\Omega \log|f_N^E|,\quad
		X_+=\ind_{\Omega_+}\log|f_N^E|.	
	\end{equation*}
	We have 
	\begin{equation*}
		\ex(\log|f_{N}^E| )\ge \ex(X)+\ex(X_-)
		\ge \frac{1}{2}(NW)^{-1/2-\epsilon}\ex(X^2)+\ex(X_-).
	\end{equation*}
	The conclusion will follow after we provide lower bounds for $ \ex(X^2) $ and $ \ex(X_-) $.
	
	We have 
	\begin{equation*}
		\ex(X^2)\ge \var(\log|f_{N}^E|)-\ex(X_-^2)-\ex(X_+^2).
	\end{equation*}	
	From \cref{thm:Cartan-W-negative-values} and 
	\cref{lem:abstract-moment-estimate} it follows that 
	\begin{equation*}
		\ex(X_-^2)\le CW^2(\log(NW))^2\le\frac{\beta_0}{4} NW,
	\end{equation*}
	provided that $ N\ge C'(D_0,D_1,|E|,\beta_0,\epsilon)W^{1+\epsilon} $.
	From the assumption that 
	\begin{equation*}
		\ex(\log|f_{N}^E| )\le 
	(NW)^{1/2}
	\end{equation*}
 	and the large deviations estimate it follows that
	\begin{equation*}
		\pr(\Omega_+)\le\pr\left(|\log|f_N^E|-\ex(\log|f_N^E|)|>(NW)^{1/2+\epsilon}\right)
		\le\exp(-(NW)^{\epsilon/2}/2)
	\end{equation*}
	for $ N\ge C(D_0,D_1,|E|,\epsilon) $. It now follows from \cref{lem:Cartan-estimates} and 
	\cref{lem:abstract-moment-estimate} that
	\begin{multline*}
		\ex(X_+^2)\le \sqrt{\pr(\Omega_+)} \sqrt{\ex((\log|f_{N}^E| )^4)}
		\le C\exp(-(NW)^{\epsilon/2}/4)(NW)^2\log^2(NW)\\
		\le \frac{\beta_0}{4}NW,
	\end{multline*}
	provided $ N\ge C'(D_0,D_1,|E|,\beta_0,\epsilon)$. We now have that 
	\begin{equation*}
		\ex(X^2)\ge\frac{\beta_0}{4}NW,	
	\end{equation*}
	for $ N\ge C(D_0,D_1,|E|,\beta_0,\epsilon)W^{1+\epsilon}  $.
	
	Another application of \cref{thm:Cartan-W-negative-values} and \cref{lem:abstract-moment-estimate} yields
	\begin{equation*}
		\ex(X_-)\ge -CW\log(NW)\ge   
		-\frac{\beta_0}{8}(NW)^{1/2-\epsilon}
	\end{equation*}
	provided $ N\ge C'(D_0,D_1,|E|,\beta_0,\epsilon) W^{1+5\epsilon} $. We used the fact that
	\begin{equation}\label{eq:W-epsilon}
		(W^{1+5\epsilon}W)^{1/2-\epsilon}\ge W^{1+\epsilon/4},
	\end{equation}
	for $ \epsilon\ll 1 $.
		
	We can now conclude that
	\begin{multline*}
		\ex(\log|f_{N}^E| )\ge \ex(X)+\ex(X_-)
		\ge \frac{1}{2}(NW)^{-1/2-\epsilon}\frac{\beta_0}{4}NW-\frac{\beta_0}{8}(NW)^{1/2-\epsilon}\\
		= \frac{\beta_0}{8} (NW)^{1/2-\epsilon},
	\end{multline*}
	provided $ N\ge C(D_0,D_1,|E|,\beta_0,\epsilon) W^{1+5\epsilon} $.
\end{proof}
% <LEM: LOWER BOUND FOR log|f_N|
We are finally able to prove \cref{thm:lower-bound-sum-Lyapunov}.
\begin{proof}(of \cref{thm:lower-bound-sum-Lyapunov})
	From \cref{prop:rate-of_convergence} and \cref{lem:lower-bound-determinant}, with 
	$ N=C W^{1+5\delta} $, $ \delta\ll 1 $, and $ C=C(D_0,D_1,|E|,\beta_0,\delta) $ large enough, we obtain
	\begin{multline*}
		\gamma_1^E+\ldots+\gamma_W^E
		\ge \frac{\beta_0}{8}\frac{(NW)^{1/2-\delta}}{N}
			-C\frac{W\log(NW)}{N}
		\ge \frac{\beta_0}{16}\frac{(NW)^{1/2-\delta}}{N}\\
		\simeq c \frac{(W^{1+5\delta}W)^{1/2-\delta}}{W^{1+5\delta}}
		\ge c W^{-5\delta}.
	\end{multline*}
	For the last inequality we used \cref{eq:W-epsilon}.
	The conclusion follows immediately.
\end{proof}
% <THM: LOWER BOUND FOR THE SUM OF LYAPUNOV EXPONENTS

% <SEC: LOWER BOUND FOR THE EXPECTED VALUE

\bibliographystyle{alpha}
\bibliography{../Schroedinger}

\vspace{2cm}

\begin{flushleft}
\textbf{I. Binder}: Dept. of Mathematics, University of Toronto, Toronto,
ON, M5S 2E4, Canada; \texttt{ilia@math.utoronto.ca}
\par\end{flushleft}

\medskip{}

\begin{flushleft}
\textbf{M. Goldstein}: Dept. of Mathematics, University of Toronto, Toronto,
ON, M5S 2E4, Canada; \texttt{gold@math.utoronto.ca}
\par\end{flushleft}

\medskip{}

\begin{flushleft}
\textbf{M. Voda}: Dept. of Mathematics, University of Toronto, Toronto,
ON, M5S 2E4, Canada; \texttt{mvoda@math.utoronto.ca}
\par\end{flushleft}

\end{document}